%% file: main_arxiv.tex
\documentclass[11pt]{article}
\usepackage{amsmath,amsthm,amssymb}
\usepackage[utf8]{inputenc}
\usepackage{fullpage}
\usepackage{hyperref}
\input{notations}

\title{Group Fairness and Multi-criteria Optimization in School Assignment}

\author{
  Santhini K.~A. \\
  Indian Institute of Technology, Madras, India\\
  GPM Government College Manjeshwar, India\\
  \texttt{cs18d013@cse.iitm.ac.in} \\
  \and
 Kamesh Munagala \thanks{Supported by NSF grants CCF-2113798 and IIS-2402823.}\\
  Duke University, USA\\
  \texttt{kamesh@cs.duke.edu} \\
  \and
  Meghana Nasre\\
  Indian Institute of Technology, Madras, India\\
  \texttt{meghana@cse.iitm.ac.in}
  \and
  Govind S. Sankar\samethanks[1] \\
  Duke University, USA \\
  \texttt{govind.subash.sankar@duke.edu} 
}
\newif\iffullversion
\fullversiontrue

\usepackage{booktabs}
\usepackage{algorithm}
\usepackage{algpseudocode}

\newtheorem{theorem}{Theorem}
\newtheorem{claim}[theorem]{Claim}
\newtheorem{lemma}[theorem]{Lemma}
 \newtheorem{corollary}[theorem]{Corollary}

\usepackage{multirow}
\begin{document}

\maketitle

\begin{abstract}
We consider the problem of assigning students to schools when students have different utilities for schools and schools have limited capacities. The students belong to demographic groups, and fairness over these groups is captured either by concave objectives, or additional constraints on the utility of the groups. We present approximation algorithms for this assignment problem with group fairness via convex program rounding. These algorithms achieve various trade-offs between capacity violation and running time.  We also show that our techniques easily extend to the setting where there are arbitrary  constraints on the feasible assignment, capturing multi-criteria optimization. We present simulation results that demonstrate that the rounding methods are practical even on large problem instances, with the empirical capacity violation being much better than the theoretical bounds.
\end{abstract}

\input{intro}
\input{model}

\input{related}
\input{results}

\input{schoolchoice}

\input{covering_main}
\input{experiments}
\input{conclusion}
\bibliography{references}
\appendix

\input{hardness}
\input{covering}

\end{document}

%% file: notations.tex
\usepackage{url}
\usepackage{graphicx}
\usepackage{soul,xcolor}
\usepackage{tikz}
\usepackage{optidef}
\usepackage{float}
\usepackage{complexity}

\usepackage{bm}
\usepackage[capitalise]{cleveref}
\usetikzlibrary{decorations.pathreplacing}
\usepackage{todonotes}
\usepackage{physics}
\usepackage{caption}
\usepackage{subcaption}

\newcommand*\samethanks[1][\value{footnote}]{\footnotemark[#1]}

\usepackage{xspace}

%% file: intro.tex
\section{Introduction}
Several societal decision-making problems manifest as assignment or matching. This includes the well-known school assignment or school redistricting problems, variants of which are implemented in several cities, including New York~\cite{abdulkadirouglu2003school, abdulkadirogu_nyc}, Boston~\cite{Peng}, and San Francisco~\cite{allman_sf_diversity}. Typically, students express preferences over schools, and schools have {\em priorities} over different types of students and a fixed capacity to accept students. This assignment can then be modeled as a matching problem over students and school seats with one or two-sided preferences. 

A standard solution approach is to find a stable
matching \cite{schoolchoice,abdulkadirouglu2003school} but often,
legislative or policy objectives require the problem to be augmented with additional features such as quotas and demographic requirements on the student body selected~\cite{abdulkadiroglu2011generalized,fragiadakis_quotas,India-match}.

While many mechanisms for school choice attempt to reconcile these policy requirements with desirable traits like strategy-proofness or stability, we turn our focus towards a different consideration. Instead of designing mechanisms to achieve these traits, we examine schools purely as resources to be allocated fairly among the students.
We consider a viewpoint from the perspective of the student demographics, such as location, race, income, or parental education level, and seek a matching that is fair to these demographic groups. 
We adopt a model of \textit{cardinal} preferences for students, where there is a numerical value to the utility that a student receives from being assigned to a school. This can capture outcomes such as average travel distance or assignment to higher-ranked schools.
In such settings, the students are not strategic. Furthermore, concerns about stability may not be relevant when schools do not have preferences, or when demographic fairness is the primary goal.
For example, as Abdulkadiro{\u{g}}lu and  S{\"o}nmez \cite{abdulkadiroglu2011generalized} note,
\begin{quote}
``During the redesign of the admissions process, BPS [Boston Public Schools] and the public considered the option of violating priorities at regular schools to promote student welfare. Likewise, the New York City high school system involves some schools at which respecting priorities emerges as a major policy goal and some other schools where priority violations may not be a cause of concern.''    
\end{quote}


Analogous to previous work on cardinal preferences \cite{allman_sf_diversity, ashlagi_cardinal_optimalallocation, shi_cardinal_optimalpriority}, we guarantee the existence of assignments that are fair to various demographic groups at the cost of adding a small number of extra seats to schools. Our model is general and allows for a variety of fairness objectives. 
We are required to assign students to schools subject to: (1) matching every student, and (2) being fair on the utilities to a pre-defined set of $g$ (potentially overlapping) demographic groups of students while (3) respecting the capacities of schools as much as possible. 
As mentioned before, these groups can capture attributes like race, location, parental income, etc. Each student has cardinal utilities over schools, and the group fairness could either be captured by an objective defined over the total utility obtained by each group, or as a set of constraints capturing the same. We note that though we have presented school assignment as a canonical application, the assignment problem we consider is very general and has many applications such as job assignment or course assignment. 

By now viewing the demographic groups as players in a fair division problem, where the schools are the resource to be allocated fairly, we can optimize for objectives such as minimum welfare or proportionality. Our main result (\cref{thm:main2}) is an algorithm to find an \textit{exactly} proportional allocation, by adding $O(g^2)$ extra seats to the schools, where $g$ is the number of demographic groups. Since $g$ is small in practice, this is a mild violation. 

Our model is similar to the one recently introduced in~\cite{PTF}.
However, we significantly extend their results at the cost of a slightly larger capacity violation. Firstly, in their model, the demographic groups are required to form a partition of the students. We have no such restriction, and prove our results under the assumption that demographic groups can overlap arbitrarily - which is the case in practice when considering unrelated features such as race, sex, nationality, etc.
Second, they assume that all students value schools in the same way. That is, there is a single, global utility function.
Such an assumption is also not based in practice, where the student preferences may be globally correlated but exhibit large individual variations.
Finally, they achieve only \textit{approximate} proportionality, and cannot even guarantee that the allocation is Pareto-optimal on utilities. In contrast, our algorithms guarantee not only \textit{exact} proportionality, but a host of other fairness measures such as maximin fairness and Nash welfare.
Compared to their $O(g\log g)$ violation in total capacities, we violate capacities by a marginally larger amount ($O(g^2)$) but solve a much larger class of problems.
To the best of our knowledge, our work is the first to achieve not just exact proportionality, but any concave fairness objective, in this setting while violating the school capacities by a small amount.

Lastly, we remark that our model can capture school-side preferences, albeit not to the extent of stability.
Suppose that a school values each student across $h$ different axes such as academics, sports, co-curriculars, etc.
Then we can write $h$ constraints for each school, to ensure that the student body assigned to the school is valued at at least (or at most)
a certain threshold in each axes. In general, our framework can model such multi-objective optimization by incurring additive violations in the capacities as well as the constraints.
In the school example, this would entail violating the capacities by an additive $O(hm)$, where $m$ is the number of schools, and a similar factor in the constraints. Since $m\ll n$, the number of students, in most practical settings, this is still acceptable.


%% file: model.tex
\subsection{Model}\label{sec:prelim}

Formally, there is a set $S$ of $n$ students divided into $g$ possibly overlapping groups $S_1, S_2, \ldots, S_g$. There is a set $T$ of $m$ schools, and school $j\in T$ has capacity $C_j$.
There is also a bipartite graph $G=(S\cup T,E)$, where $(i,j)\in E$ is an edge if for $i\in S, j\in T$, it is possible to assign student $i$ to school $j$.
Let $\vec{y}$ denote an assignment, where $y_{e} \in \{0,1\}$ denotes whether for edge $e=(i,j)\in E$, student $i \in S$ is assigned to school $j \in T$. This assignment is \textit{feasible} if it satisfies the capacity constraint of each school, and each student is assigned to some school\footnote{Note that this is without loss of generality, since we can always add a dummy school with infinite capacity to ensure that every student can be matched. However, this naturally affects the fairness objective. For example, the proportional share of demographic groups can increase from this process. }.
Let $u_{ij}$ be the (non-negative) utility derived by student $i$, if assigned to school $j$. 

Given an assignment $\vec{y} \in \mathcal{P}$, we define the utility of each group $S_k$ as $U_k(\vec{y}) = \sum_{i \in S_k} \sum_{(i,j)\in E} u_{ij} y_{(i,j)}.$ Let $\vec{U} = \langle U_1, U_2, \ldots, U_g \rangle$. 
The goal is to find a feasible assignment $\vec{y}$  that maximizes some fairness function on the utilities perceived by the $g$ groups. Let $f(\cdot)$ be a non-decreasing concave function. Then, a general goal is to maximize
$$ h(\vec{U}) = \sum_{k=1}^g f\left(U_k(\vec{y})\right).$$

In practice, several functions $f$ can be reasonable. For instance, the celebrated {\em Nash welfare} objective sets $f = \log$, and the optimal solution $\vec{U^*}$ satisfies the following relation: For any other feasible utility vector $\vec{U'}$, we have the relation:
\begin{align}
\label{eq:prop}
  &\hfill &  \frac{1}{g} \sum_{k=1}^g \frac{U'_k}{U^*_k} \le 1.
\end{align}     
The utilities $\vec{U^*}$ in this allocation are  proportional for each of the groups, that is, $U^*_k \ge  \frac{U'_k}{g}$, meaning each group gets at least $1/g$ of the utility it would have obtained had it been the only group in the system and the social welfare was maximized. This notion of proportionality is the objective considered in~\cite{PTF}. Note that Nash welfare is a much more general objective, capturing proportionality for various subsets of groups, and the approach in~\cite{PTF} does not extend to Nash welfare.  

Other examples of utility functions that we can model  include the \textbf{max-min} fairness objective, which maximizes 
$$\min_{k=1}^g U_k$$
and tries to make the least happy group as happy as possible.
This can be modeled with simple linear constraints without needing the use of a convex function.
Any objective over the utilities that has \textbf{Constant Elasticity of Substitution} (CES) with certain parameters also falls in this model. Such functions can be modeled as maximizing
$$
\sum_{k=1}^g a_k (U_k)^r.
$$
for non-negative real numbers $a_k$.
When $r\leq 1$, this function is concave and can be maximized in our framework. Note that the $r>1$ case models an inherently \textit{unfair} allocation, since it is always better to allocate more utility to the best-off group.

Finally, we can also capture group fairness through arbitrary covering or packing constraints where we are explicitly given utility requirements for each group. The objective is to find an assignment where each constraint is satisfied.
This can capture general \textit{multi-criteria optimization} for assignment problems in a non-fairness context. 
Formally, given $g$ arbitrary constraints,
we obtain a solution that can violate the capacities by a small function of $g$, and satisfies the constraints up to an additive function of $g$ and $\abs{q_{\max}}$, the largest magnitude coefficient in the constraint matrix.
We discuss this further below.

%% file: results.tex
\subsection{Our Results} \label{sec:results}

Our main contribution is a set of polynomial-time approximation algorithms for this problem for arbitrary fairness objectives. In \cref{sec:proof}, we present two approximation algorithms based on rounding a natural convex programming relaxation, which yield somewhat different guarantees, as summarized in the theorem below.

\begin{theorem}
\label{thm:main}
Given any monotone, concave fairness function $f$, let $\vec{U^*}$ be the utilities in the optimal solution. Then, there exist algorithms to compute an assignment $\vec{y}$ that satisfies relaxed school capacities $\vec{C'}$ and yields utilities $\vec{U'}$ with $U'_k \ge U^*_k$ for all groups $k \in \{1,2,\ldots, g\}$, with one of the following guarantees:

\begin{enumerate}
\item A polynomial\footnote{Throughout the paper, we use this to mean polynomial in $n,m,$ and $g$.} running time and satisfies $C'_j \le C_j + 1 + \delta_j$, where $\sum_j \delta_j \le 2g$.

\item A $n^{O(g)}$ running time and satisfies $C'_j \le C_j + \delta_j$ with $\sum_j \delta_j = O(g^2)$.
\end{enumerate}
\end{theorem}

Note that the latter algorithm is slower but yields better violation of capacities if $g \ll |T|$. For practical school choice scenarios~\cite{abdulkadirouglu2003school}, the number of students significantly outweighs the number of available schools, while the number of groups $g$ is typically small, even constant. The above violations are, therefore, quite mild.  We substantiate this via our empirical results, which we discuss later.

We also show that the school assignment problem mentioned above is {\sc NP-Hard} for the max-min fairness objective, even when the number of schools or the number of groups is only two. The latter result extends to the proportionality objective. This motivates the need to relax the capacity constraints (as in \cref{thm:main}) if our goal is to achieve a polynomial time algorithm.

\begin{theorem}\label{thm:hardness_1}
    Suppose the number of groups $g$ is part of the input, and the objective is to decide if the minimum utility received by any group is at least one. Then, the school assignment problem is {\sc NP-Complete} even when there are only two schools.
\end{theorem}

\begin{theorem}\label{thm:hardness-2}
    Suppose the objective is to decide if an exactly proportional allocation exists. Then, the school assignment problem is weakly {\sc NP-Complete}, even when the number of groups $g = 2$.
\end{theorem}
We defer the proofs to \cref{app:hard}.
We also note that to achieve the proportionality objective, it is known that a capacity violation of $\frac{g}{2}$ is required in the worst-case~\cite{PTF}.

\newcommand{\qmax}{q_{\mathrm{max}}}

\paragraph*{Packing/Covering Constraints and Multi-criteria Optimization.} In \cref{sec:covering}, we present an extension of our framework to handle assignments with more general constraints. As an application, suppose the utility of a student for a school is multi-dimensional, capturing aspects like academic excellence, or location, or diversity of student body.  The goal is to achieve at least a specified total utility value in each dimension. Such {\em multi-objective} optimization \cite{gourves_multiobjective,serafini_multiobjective,shimada2020multi} can be modeled by covering or packing constraints, and we present a result similar to \cref{thm:main} for this general setting. Formally, we solve a linear relaxation of the following integer program, where we have a setting as in \cref{sec:prelim} but instead of the fairness objective and utilities, we have a matrix $Q\in \mathbb{R}^{n\times r}$ defining $r$ linear constraints that we are required to satisfy.
In other words, we wish to solve the following integer program:

\begin{minipage}{0.1\linewidth}
\begin{equation}
    \tag{IP}\label{IP}
\end{equation}
\end{minipage}
\begin{minipage}{0.8\linewidth}
    \begin{alignat}{2}
    \sum_{j} y_{ij} & = 1 && \forall \mbox{ students } i \label{lp1:distribution}\\
    \sum_{i } y_{ij} & \le  C_j && \forall \mbox{ schools } j \label{lp1:school}\\
    \sum_{i,j} q_{ij}^{\ell} y_{ij} &\geq Q_{\ell} && \forall \ell \in \{1,2,\ldots,r\} \label{lp1:extra-constraints}\\
    y_{ij} & \in  \{0,1\} \qquad&& \forall~i,j \label{lp1:integers}
\end{alignat}
\end{minipage}
\medskip

\newcommand{\lp}{LP\xspace}


We show that \cref{thm:main} generalizes to this setting at the cost of incurring an additional additive loss proportional to $\abs{\qmax} := \max_{i,j,\ell} \abs{q_{ij}^{\ell}}$ and $r$, the number of rows in $Q$. 

\begin{theorem}\label{thm:arbitrary}
For arbitrary constraints, when the linear programming relaxation of (\ref{IP}) has a feasible solution, there are algorithms that output an integer solution
\begin{itemize}
\item In polynomial time, such that the constraints \cref{lp1:extra-constraints} are preserved up to an additive $r\cdot \abs{\qmax}$; and if each school is given one unit extra capacity, the total violation in \cref{lp1:school} over this is $2r$. 
\item In $n^{O(r)}$ time, such that the constraints \cref{lp1:extra-constraints} are preserved up to an additive $O(r^2\cdot \abs{\qmax})$; and the total violation in \cref{lp1:school} is $O(r^2)$.
\end{itemize}
\end{theorem}

As a specific application, in \cref{sec:dominance} we study the assignment with ranks problem first considered in~\cite{asn22_optimalmatchings}. Here, each student ordinally ranks the schools with possible ties. An input signature $\vec{\rho}$ of length $r$, the goal is to find an assignment where the number of students who are assigned their first $k$ choices (for all $k \le r$) is at least $\sum_{j=1}^k \rho_j$. 
In comparison to the  algorithm in~\cite{asn22_optimalmatchings} that runs in time $n^{O(r^2)}$ and uses multivariate polynomial interpolation, our algorithms present substantial improvements in both runtime and ease of implementation at the cost of a small capacity violation, when students are given a small choice $r$ of ranks. We show that this problem has an additional `monotonicity' in the constraints, that enable us to avoid the additive violation mentioned above.

\begin{theorem}\label{thm:ranking}
Given a feasible fractional solution to a matching with ranking instance with input signature $\rho$, there is an algorithm to output a matching with signature $\sigma \succ \rho$ that satisfies relaxed school capacities $\vec{C'}$ with one of the following guarantees:
\begin{enumerate}
\item A $\poly(n,r)$ running time and satisfies $C'_j\leq C_j + 1 + \delta_j$, and $\sum_j \delta_j \le 2r$.
\item A $n^{O(r)}$ running time and satisfies  $C'_j\leq C_j+ \delta_j$, and $\sum_j \delta_j = O(r^2)$.
\end{enumerate}
\end{theorem}

\paragraph*{Benchmark and Simulation.} In \cref{sec:expt}, we present a benchmark for the group fairness objective, where we use an ILP to find the optimum violation in capacity needed to achieve the utilities generated by the convex program. \cref{thm:main} yields a theoretical upper bound on the capacity violation. However, we show that both the ILP benchmark as well as our rounding algorithms yield substantially better violations on realistic instances, hence showcasing the practicality of our approach. We also present empirical results for the aforementioned matching with ranks problem in \cref{sec:expt2}.

\subsection{Techniques and Related Work.}\label{sec:related}

Our algorithm uses LP rounding and borrows ideas from the seminal Generalized Assignment Problem (GAP) rounding technique of Lenstra, Shmoys, and Tardos~\cite{LST,ShmoysT93}. Their iterative rounding procedure involves the observation that the number of fractional variables in a vertex solution to a linear programming relaxation is bounded. We build on this idea and apply it to a linear program written on paths and cycles instead of assignments, enabling us to combine it with a theorem of Stromquist and Woodall~\cite{STROMQUIST}. This approach was recently used in a similar model by the authors in \cite{PTF}, which they called ``cake frosting''. This theorem is a consequence of the celebrated ham-sandwich theorem~\cite{HamSandwich}, and is, hence, non-constructive. Using this technique, they achieve approximate proportionality while violating the total capacity by $O(g \log g)$, where $g$ is the number of groups.

In contrast, we use convex programming relaxation to handle {\em arbitrary} fairness objectives such as proportional fairness, Pareto-optimality, and maximin fairness, vastly generalizing the space of objectives. Our method only loses $O(g^2)$ on the total capacity, while preserving utilities from the fractional relaxation. For instance, our method would achieve {\em exact} proportionality, and by \cref{eq:prop}, it even achieves a generalization of this concept to subsets of groups. Further, we show empirically that the use of convex programming keeps the cake frosting instance very small and hence tractable.

At a high level, our main technical contribution is to show how the cake frosting method can be applied to certain types of fractional solutions, in particular, a vertex solution constructed via GAP rounding of the convex programming relaxation.  We hence showcase the full power of the technique in~\cite{STROMQUIST}.  In addition, as discussed above, our techniques extend smoothly to handle arbitrary covering or packing constraints on the allocations, which is motivated by multi-objective optimization and rank optimization.

We note that the idea of using cake frosting to round fractional solutions has appeared before for {\em packing problems} in Grandoni {\em et al.}~\cite{GrandoniRSZ14}, where the authors develop a PTAS for matchings in general graphs with $O(1)$ budget constraints on the set of chosen edges. At a high level, all these approaches -- the ones in~\cite{GrandoniRSZ14,PTF} and our work -- apply cake frosting to decompose paths and cycles to approximately preserve constraints, but differ in the details of how the paths and cycles are constructed from the integer or fractional solutions. For instance, in contrast to~\cite{PTF}, which defines the frosting function based on schools, we define it based on students, hence avoiding an additive violation on the utility. Further, since~\cite{GrandoniRSZ14} consider packing problems, their reduction to cake frosting is entirely different in the technical details.

\paragraph*{Matching with Violations.}
Unlike many resource allocation problems, school assignments have flexibility in the capacities assigned to schools. Additional seats can be added with appropriate investments, or minor adjustments to class structures.
Governments have also shown a willingness to add seats, particularly in situations where students would have gone unassigned \cite{bobbio_schoolchoice}.
Along these lines, several papers have considered such assignment problems with small capacity violations. These papers mostly fall into two categories - those that try to directly optimize the capacity violations in some form while achieving a set goal like stability or perfectness~\cite{bobbio2022capacity,chen23_optimalviolation,vaish_stability_capacity,asn22_optimalmatchings} and those that optimize some other objective like fairness with provably small capacity violations~\cite{GrandoniRSZ14,nguyen2018near,PTF}.
Our model falls in the latter category -- we wish to find an assignment that satisfies some notion of fairness while violating capacities by as little as possible.


\paragraph*{Group Fairness.} The school assignment problem that we study was first considered recently in Procaccia, Robinson, and Tucker-Foltz~\cite{PTF}. The only objective considered in this work is {\em proportionality} --- in the assignment, each of the $g$ demographic groups is required to achieve at least $1/g$ fraction of the utility it could have achieved had it been the only group in the system.
Various such notions of group fairness have been studied in many contexts such as clustering \cite{bera2019fair,ghadiri21_sociallyfair}, knapsack \cite{patel2020group}, and matchings \cite{esmaeili2023rawlsian,sankar2021matchings}.
The objective function in \textit{Socially Fair $k$-clustering} \cite{ghadiri21_sociallyfair}, where the average clustering cost across each demographic group has to be minimized, is particularly similar to ours.

%% file: schoolchoice.tex
\section{Approximation Algorithm: Proof of \cref{thm:main}}
\label{sec:proof}
In this section, we prove \cref{thm:main}. We begin with a convex programming relaxation to the problem and then present two rounding schemes that yield the two guarantees in the theorem. 

\subsection{Convex Program Relaxation} 
\label{sec:convex}
Recall that there is a set $S$ of $n$ students divided into $g$ possibly overlapping groups sets $S_1, S_2, \ldots, S_g$. There is a set $T$ of schools, where school $j$ has capacity $C_j$.
Finally, there is a bipartite graph $G=(S\cup T,E)$ between the students and schools that represents possible assignments.
An assignment of students to schools is a feasible solution $\vec{y} \in \mathcal{Q}$, where $\mathcal{Q}$ is the polytope defined by the following constraints:
\[ \begin{array}{rcll}
\sum_{j \in T} y_{ij} & = & 1 & \forall i \in S \\
\sum_{i \in S} y_{ij} & \le & C_j & \forall j \in T \\
y_{ij} & \in & \{0,1\} & \forall (i,j) \in E
\end{array}
\]
Let $\mathcal{P}$ denote the linear relaxation of $\mathcal{Q}$, where the last constraint is replaced by $0\leq y_{ij}\leq 1$. The first step is to write the following convex programming relaxation:
\[ \mbox{Maximize} \sum_{k=1}^g f(U_k) \]
\[ \begin{array}{rcll}
 \sum_{i \in S_k} \sum_{j \in T} u_{ij} y_{ij} & \ge & U_k & \forall  \mbox{ groups } k \\
\vec{y} &\in & \mathcal{P}\\
U_k & \ge & 0 & \forall \mbox{ groups } k
\end{array} \]

This can be solved in polynomial time. Let the optimal solution to the convex program yield utility vector $\vec{U^*}$. We now need to round the following set of constraints so that the $\{y_{ij}\}$ values are integer. Denote this formulation as (LP1).
\[ \begin{array}{rcll}
  \qquad  \vec{y} \in \mathcal{P},\qquad \forall  \mbox{ groups }k, \sum_{i \in S_k} \sum_{j \in T} u_{ij} y_{ij} & \ge & U^*_k & .
\end{array}
\]
We now present two rounding algorithms that yield the corresponding guarantees in \cref{thm:main}.

\subsection{Generalized Assignment Rounding}
\label{sec:gap}

The first rounding algorithm is similar to rounding for generalized assignment (GAP)~\cite{LST,ShmoysT93} and is presented in \cref{alg:approx}.
We remark that the iterative procedure from Step~\ref{alg1-step:begin-iteration}-\ref{alg1-step:end-iteration} is not necessary; the same can be achieved with a single LP solution.
We present it this way for ease of exposition.
We show that it achieves the following guarantee, yielding the first part of \cref{thm:main}.

\begin{theorem}
\label{thm:main1}
\Cref{alg:approx} runs in polynomial time and finds an integer solution $\vec{y} \in \mathcal{P}$ that satisfies relaxed school capacities $\vec{C'}$ and yields utilities $\vec{U'}$, where
\begin{enumerate}  
\item $U'_k \ge U^*_k$ for all groups $k$, and
\item $C'_j \le C_j + 1 + \delta_j$, where $\sum_j \delta_j \le 2g$.
\end{enumerate}
\end{theorem}

\begin{algorithm}[htbp]
\caption{GAP Rounding}\label{alg:approx}
\hspace*{\algorithmicindent}
\begin{algorithmic}[1]
\Repeat \label{alg1-step:begin-iteration}
\State Obtain a vertex solution $\vec{y}$ to (LP1).
\For{all $y_{ij}=b\in \{0,1\}$}
    \State Fix $y_{ij}=b$ and remove this variable from (LP1), and update the constraints.\label{alg1-step:integral}
\EndFor
\Until{(LP1) is not modified.}\label{alg1-step:end-iteration}
\State For each remaining student $i$ (these have a degree more than $1$), assign this student to $\mbox{argmax}_j \{u_{ij} | y_{ij} > 0 \}$. \label{alg1-step:round}
\end{algorithmic}
\end{algorithm}

\begin{proof}[Proof of \cref{thm:main1}.] We denote the number of incident edges with $y_{ij} > 0$ as the ``degree'' of a vertex. First, note that since each student $i$ is assigned to $\mbox{argmax}_j \{u_{ij} | y_{ij} > 0 \}$,  the utility in the integer solution is at least that in (LP1).

We now bound the capacity violation. We note that Step~\ref{alg1-step:integral} cannot violate any capacities since $\Vec{y}$ was feasible for (LP1). 
It  remains to argue that Step~\ref{alg1-step:round} does not incur too many capacity violations.
Let $E$ be the set of remaining edges with $y_{ij} \in (0,1)$.  Since $\Vec{y}$ was an extreme point solution to (LP1), $|E|$ constraints of (LP1) must be tight. 

At the beginning of Step~\ref{alg1-step:round}, let $T'$ be the set of remaining schools and among these, let $\hat{T}$ be those whose capacity constraints are tight.  Let $S'$ be the set of remaining students. Each student has a tight constraint associated with it. Suppose $g'$ of the $g$ constraints corresponding to the groups are tight. Since we have a vertex solution, 
\begin{align}
    2|E| = 2(|S'|+|\hat{T}|+g')\label{eqn:vertex}
\end{align}
From the Handshaking lemma, we also have
    $\sum_{v\in S'\cup T'} \mathrm{deg}(v) = 2|E|.$
Combining this with \cref{eqn:vertex}, we have,
\begin{align*}
    \sum_{v\in S'\cup T'}\mathrm{deg}(v) - 2|S'|-2|\hat{T}|= 2g'
    \implies\sum_{v\in T'\setminus \hat{T}}\mathrm{deg}(v) +\sum_{v\in S'\cup \hat{T}} (\mathrm{deg}(v)-2) =2g'\leq 2g
\end{align*}

We know that each school in $\hat{T}$ has a degree of at least 2, since the capacity is an integer, and all assignment variables are strict fractions. Similarly,  each student in $S'$ has a degree at least $2$.  Therefore, every term in the above summation is non-negative. Let each student or school $v \in \hat{T}\cup S'$ have degree $2 + \delta_v$, while schools $v \in T' \setminus \hat{T}$  have degree $\delta_v$.
We will refer to these $\delta_v$ terms as the excess degrees.
Then, the above implies 
\begin{equation}
    \label{eq:2g} \sum_{v \in T' \cup S'} \delta_v \le 2g. 
\end{equation}

To bound the capacity violation in \cref{alg1-step:round}, we observe the following properties: First, if a school in $\hat{T}$ had degree 2, then it must have had a capacity of at least 1, and in the worst case, both students with edges to it will match to it. This leads to a violation of 1 in this school's capacity. Next, for any other school $v\in \hat{T}$, it again has capacity at least 1 and has $2+\delta_v$ students applying to it. In the worst case, this leads to a capacity violation of at most $1+\delta_v$. Finally, for  schools $v\in T'\setminus \hat{T}$, since the degree is $\delta_v$, this leads to a violation of at most $\delta_v$.

In total, this leads to a capacity violation of one per school and the excess degrees $\sum_v \delta_v$ lead to an additional $2g$ violation overall. This completes the proof.
\end{proof}

\subsection{Improved Capacity Violation via Cake Frosting}\label{sec:cake-cutting}

While the previous section provides a polynomial-time solution, we can improve the capacity violation bound, albeit at the cost of increased runtime. We achieve this by replacing the last step in \cref{alg:approx} with a more sophisticated 'cake frosting' technique, building on the work of~\cite{PTF}.
We show the following theorem, corresponding to the second part of \cref{thm:main}. 

\begin{theorem}
\label{thm:main2}
There is a  $n^{O(g)}$ time algorithm that computes an integer assignment $\vec{y} \in \mathcal{P}$ that satisfies relaxed school capacities $\vec{C'}$ and yields utilities $\vec{U'}$, such that:
\begin{enumerate}
    \item $U'_k \ge U^*_k$ for all groups $k$; and
    \item $C'_j \le C_j + \delta_j$ with $\sum_j \delta_j = O(g^2)$.
\end{enumerate}
\end{theorem}

\paragraph*{Paths and Cycles.}
Let $G$ be the graph at the beginning of Step~\ref{alg1-step:round} in \cref{alg:approx}. Recall that the maximum degree in $G$ was 2, except for some vertices with excess degrees in \cref{eq:2g}. We will process the graph into a graph of maximum degree 2, with some additional properties, in \cref{alg:createpath}.  

\begin{algorithm}[htbp]
\caption{Graph Modification}\label{alg:createpath}
\hspace*{\algorithmicindent}
\begin{algorithmic}[1]
 \For {each student $i$ with degree strictly more than $2$} 
 \State Add a capacity of one to $j^* = \mbox{argmax}_j \{u_{ij} | y_{ij} > 0 \}.$
 \State Fix $y_{ij^*}=1$ and remove this student. \label{alg2-step:maxutil}
 \EndFor
 \State $ S_1 = \{ j | j \in \hat{T}, \mbox{degree}(j) > 2 \}$.
 \State $S_2 = \{ j | j \in \hat{T}, \mbox{degree}(j) = 2, C_j = 2\}$.
\For {each school $j \in S_1 \cup S_2 \cup (T' \setminus \hat{T})$}
\State $d = \mbox{degree}(j)$.
\State Create $d $ copies of $j$, each with capacity one. 
\State Assign (add an edge from) each $i$ with $y_{ij} > 0$ to a distinct copy of $j$.
\State \Comment{Each new school has degree one.}
\EndFor
\State If a school has degree one, reduce its capacity to one.
\end{algorithmic}
\end{algorithm}

At the end of the process, let $G(V,E)$ be the resulting graph on fractional edges. Any vertex has a degree of at most two, and hence we get a graph with the following structure: Every connected component is a path or a cycle; every student has degree exactly two, and is an internal node of a path or cycle; every school $j \in T'$ has capacity one and degree at most two; and finally, any school $j \in T' \setminus \hat{T}$ has degree one and capacity one, and is, therefore, a leaf of a path.

\begin{lemma}
\cref{alg:createpath} violates the total capacity by an additional $4g$.
\end{lemma}
\begin{proof}
For $j \in S_1$, let $\mbox{degree}(j) = 2 + \delta_j > 2$, implying $\delta_j \ge 1$.  We increase the capacity by $1 + \delta_j \le 2 \delta_j$. For $j \in S_2$, the new capacity is the same as the original capacity. For $j \in T' \setminus \hat{T}$, suppose $\mbox{degree}(j) = \delta_j$, then we increase the capacity by $\delta_j$. By \cref{eq:2g}, the total increase is at most $4g$.
\end{proof}

In the graph $G$, suppose $e = (i,j)$; then we denote $x_e = y_{ij}$. This graph is a collection of paths and cycles. For non-leaf school or student $v$, the above conditions imply $x_{e_1} + x_{e_2} = 1$ if the two edges incident on $v$ are $e_1$ and $e_2$. This follows because a degree-two school must belong to $\hat{T}$ and corresponds to a tight constraint, and any student is associated with a tight constraint. This implies the following claim:
\begin{claim}
For every component (path or cycle) $C$ of $G$, there is some $\alpha\in (0,1)$ such that every even edge $e$ in the component has $x_e=\alpha$ and every odd edge has $x_e=1-\alpha$.
\end{claim}
\paragraph*{Bounding the Number of Fractional Components.}
We view this fractional solution as follows. For component $C$, set $z_C=\alpha$ if $x_e=\alpha$ for every even edge. Let $u_C^{even}(i)$ be the utility that group $i$ gets in the assignment that selects all even edges (and no odd edges) of component $C$. Let $u_C^{odd}(\ell)$ be the utility that group $\ell$ gets in the assignment that selects all odd edges of component $C$. We modify (LP1) to the following, where $\hat{U}^*_{\ell}$ is the modified utility after removing the integral variables. 
\[ \begin{array}{lcl}
 \forall \mbox{ groups } \ell,& &
 \qquad \displaystyle\sum_{C} z_C \cdot u_C^{even}(\ell)+(1-z_C) \cdot u_C^{odd}(\ell)  \ge \hat{U}^*_{\ell} \\
\forall \mbox{ components } C,& &
\hfill z_C \in [0,1].
\end{array}
\]
Let $s$ denote the number of variables. In any extreme point solution, at least $s-g$ of the constraints $z_C \in [0,1]$ are tight, which means that at most $g$ of the $z_C$ variables can be fractional, in $(0,1)$. For all integral $z_C$, we select the even matching if $z_C=1$ or the odd matching if $z_C = 0$. Remove these variables and rewrite the above LP just on the fractional variables. 

\begin{algorithm}[htbp]
\caption{Cake Frosting Rounding}\label{alg:frost}
\hspace*{\algorithmicindent}
\begin{algorithmic}[1]
\For{every student $i$}
\If {$[\frac{i-1}{r}, \frac{i}{r})\subseteq X$} 
\State Choose the edge from the even matching for student $i$, and include $i$ in set $T_1$.
\ElsIf {$[\frac{i-1}{r}, \frac{i}{r})\subseteq [0,1]\setminus X$} 
\State Choose the edge from the odd matching for student $i$ and include $i$ in set $T_2$.
\Else
\State Assign $i$ to $\mbox{argmax}_{j'} \{u_{ij'}, y_{ij'} > 0\}$. \label{alg3-step:maxutil}
\EndIf
\EndFor
\end{algorithmic}
\end{algorithm}

\paragraph*{Reduction to Cake Frosting.}
For the $g$ components with fractional $z_C$, we need to find an integral solution that approximately preserves the utilities.
This would be achieved if we could `interpolate' $z_C$ fraction from the odd matching to the even matching.
We can view this as a cake-frosting problem as in~\cite{PTF}, where the $g$ groups are the players. First, we convert each cycle into a path as follows: Pick some student $i$ on this cycle, assign $i$ to $\mbox{argmax}_j \{u_{ij} | y_{ij} > 0\}$, and delete this student. This step increases the capacity of at most $g$ schools by one and reduces each cycle to a path that begins and ends at a school.
We now present a generalization of the ``cake frosting'' lemma first presented in~\cite{STROMQUIST} and used in~\cite{GrandoniRSZ14,PTF}.

\begin{lemma}[Cake Frosting Lemma]
Given $g$ piecewise constant functions $f_{\ell}, \ell = 1,2,\ldots, g$ with domain $[0,1]$, and given any $\alpha \in (0,1)$, there is a `perfect frosting' $X \subseteq [0,1]$ written as a union of at most $2g-1$ intervals such that for all $\ell$:
$$ \int_X f_{\ell}(x) dx = \alpha \cdot \int_0^1 f_{\ell}(x) dx.$$
\end{lemma}

We now show how to apply the above lemma similarly to~\cite{PTF,GrandoniRSZ14}. Fix a path $C$. Let $z_C=\alpha$. Let there be $r$ students in $C$, indexed from $1$ to $r$. We divide the interval $[0,1]$ into $r$ parts where $[\frac{i-1}{r}, \frac{i}{r})$ belongs to the $i^{th}$ student.\footnote{This is in contrast to the method in~\cite{PTF}, which defines intervals based on schools.} Define for every group $\ell$,
\begin{align*}
    u_{even}(\ell,i) &= \begin{cases}
        u_{ij} & \mbox{if the even matching assigns student $i$ }\\
        &\mbox{to school $j$ and student $i$ is in group $\ell$}\\
        0 & \mbox{Otherwise}
    \end{cases}\\
    u_{odd}(\ell,i) &= \begin{cases}
        u_{ij} & \mbox{if the odd matching assigns student $i$}\\
        &\mbox{to school $j$ and student $i$ is in group $\ell$}\\
        0 & \mbox{Otherwise}
    \end{cases}
\end{align*}

For $x\in [\frac{i-1}{r}, \frac{i}{r})$, define $f_\ell(x)= r(u_{even}(\ell,i)-u_{odd}(\ell,i)).$


\paragraph*{Rounding Procedure.}  For path $C$, we now apply the cake frosting lemma to the function $f$ as defined above, with $\alpha = z_C$ to find the perfect frosting $X$ that is a union of at most $2g-1$ intervals. Given $X$, we construct the assignment as in \cref{alg:frost}.
The final algorithm applies this procedure separately to each of the $g$ fractional paths. Note that the $\alpha$ value depends on the path.

\paragraph*{Analysis.} We first bound the utility of each group $\ell$ in path $C$.
Define $T_3:=[r]\setminus(T_1\cup T_2)$, i.e. the set of students in $C$ not in $T_1$ or $T_2$.
The utility of group $\ell$ in the solution is
\begin{align*}
    &\sum_{i\in T_1} u_{even}(\ell,i)+\sum_{i\in T_2} u_{odd}(\ell,i)+\sum_{i\in T_3} \max(u_{even}(\ell,i),u_{odd}(\ell,i))\\
    =&\sum_{i\in T_1} (u_{even}(\ell,j)-u_{odd}(\ell,j))
    +\sum_{i\in [r]} u_{odd}(\ell,j)  +\sum_{i\in T_3} \max(u_{even}(\ell,i),u_{odd}(\ell,i))-u_{odd}(\ell,i)\\
        \geq &\sum_{i\in T_1} (u_{even}(\ell,j)-u_{odd}(\ell,j))
    +\sum_{i\in [r]} u_{odd}(\ell,j)
    +\sum_{i\in T_3} \abs{\left[\frac{i-1}{r},\frac{i}{r}\right)\cap X }(u_{even}(\ell,i)-u_{odd}(\ell,i))
\end{align*}
\begin{align*}
    = &\frac{1}{r}\sum_{i\in T_1}\int_{x\in [\frac{i-1}{r},\frac{i}{r})} f_\ell(x)  +u_C^{odd}(\ell) + \frac{1}{r}\sum_{i\in T_3}\int_{x\in [\frac{i-1}{r},\frac{i}{r})\cap X} f_\ell(x)\\
    = &\frac{1}{r}\int_{x\in X} f_\ell(x)  +u_C^{odd}(\ell)
    = \frac{\alpha}{r} \cdot \int_{x\in [0,1]} f_\ell(x)  +u_C^{odd}(\ell)\\
    = &\alpha \cdot (u_C^{even}(\ell)-u_C^{odd}(\ell))+u_C^{odd}(\ell)    =\alpha \cdot u_C^{even}(\ell)+(1-\alpha) \cdot u_C^{odd}(\ell).
\end{align*}
The first equality follows by adding and subtracting $\sum_{i\in T_1\cup T_3}u_{odd}(\ell,i)$.
The second line and the only inequality follows from the observation that $\abs{\left[\frac{i-1}{r},\frac{i}{r}\right)\cap X }\leq 1$ and $\max(u_{even}(\ell,i),u_{odd}(\ell,i))-u_{odd}(\ell,i)\geq 0$.
The third line follows from the definition of $f$,
the fourth line follows from the the structure of $X$ and $T_1,T_3$, and the fifth follows from the cake frosting lemma. The above chain of inequalities shows that for each group $\ell$, the integer solution has utility  at least that of the fractional solution.

To bound the total capacity violation, note that \cref{alg:frost} violates the capacity by one at every interval boundary. By the Cake Frosting lemma, this is an additional violation of $O(g)$ per path, and hence $O(g^2)$ overall.   This completes the proof of \cref{thm:main2}, and hence \cref{thm:main}.

%% file: covering_main.tex
\section{Generalization to Arbitrary Constraints}
\label{sec:covering}
\newcommand{\fullmatch}{\textsc{Matching with Solvable Constraints}\xspace}
\newcommand{\match}{MSC\xspace}
\newcommand{\generalmatch}{\textsc{Matching with Constraints}\xspace}
\newcommand{\Q}{Q}

We now consider a more general setting.
As before, we are given a set $T$ of schools, where school $j$ has capacity $C_j$, a set $S$ of students, and a bipartite graph $G=(S\cup T,E)$ between the students and schools. The objective is to find an integral assignment $\Vec{y}$ of all the students that satisfies an additional set of $r$ covering or packing constraints (possibly with negative coefficients\footnote{Note that the only place we require non-negativity in the coefficients is in solving the convex program.}). 
Define \lp to be the linear relaxation of (\ref{IP}) in \cref{sec:results} obtained by relaxing \cref{lp1:integers} to $y_{ij}\in[0,1]$. Unlike the previous section, a given $y_{ij}$ variable can appear in the constraints arbitrarily. 
We now show that both \cref{thm:main1,thm:main2} generalize to this setting, completing the proof of \cref{thm:arbitrary}.

\subsection{Generalizing \cref{thm:main1}}

Our algorithm runs in the following steps, which build on \cref{alg:approx}.

\begin{enumerate}
\item Solve the linear programming relaxation \lp, fix and remove the integral variables, and find a vertex solution. Let $E$ be the set of fractional variables, and $S'$ be the remaining students.
\item  Rewrite \lp on the variables $E$ and without the capacity constraints \cref{lp1:school}. 
\item Keep eliminating integer variables, stopping at a vertex solution where all variables are fractional. Let $E'$ be the remaining variables and $S''$ be the remaining students.
\item Set an arbitrary $y_{ij} > 0$ to $1$ for each $i \in S''$.
\end{enumerate}

\begin{theorem}
For arbitrary covering or packing constraints, when the linear programming relaxation has a feasible solution, there is a polynomial time algorithm that outputs an integer matching and that achieves the following guarantee: 
\begin{itemize}
\item The constraints \cref{lp1:extra-constraints} are preserved up to an additive $r\cdot \qmax$; and
\item If each school is given one unit extra capacity, the total violation in \cref{lp1:school} over this is $2r$. 
\end{itemize}
\end{theorem}
\begin{proof}
  First, the proof of \cref{thm:main1} shows that regardless of how the students in $S'$ are assigned, if each school is given one extra unit of capacity, then the total violation in capacity is at most $2r$.
    
    Therefore, we can focus on assigning the students so that the constraints \cref{lp1:extra-constraints} are not violated significantly.    
    In Step (2), since any student in $S''$ has degree at least $2$, we have $|E'| \ge 2|S''|$. Further, any extreme point in Step (3) has exactly $|E'|$ tight constraints. Since the number of potential tight constraints is at most $|S''| + r$, we obtain $|S''| \le r$. Therefore Step (4) violates each constraint by an additive $r\cdot \qmax$, completing the proof.
\end{proof}

\subsection{Generalizing \cref{thm:main2}}
We next generalize \cref{thm:main2}. We first apply \cref{alg:approx} to the linear programming relaxation, stopping before Step~\ref{alg1-step:round}. We then follow the procedure in \cref{sec:cake-cutting} and sequentially apply \cref{alg:createpath,alg:frost} to the fractional solution. To set up the cake frosting game to apply \cref{alg:frost}, we view each of the $r$ constraints in \cref{lp1:extra-constraints} as a player of the cake frosting instance. Define $u_{even}(\ell,i):=q_{ij}^{\ell}$ where $(i,j)$ is the even matching edge adjacent to $j$ and define $u_{odd}(\ell,i)$ similarly. The only steps that are different are the assignment steps -- Step~\ref{alg2-step:maxutil} in \cref{alg:createpath} and Step~\ref{alg3-step:maxutil} in \cref{alg:frost}. Here, we perform an arbitrary assignment of the students to the schools. We present all the details in \cref{alg:overall-general} for completeness.

\iffullversion

\begin{algorithm}[htbp]
\caption{Algorithm for \cref{thm:generalization-main2}}\label{alg:overall-general}
\hspace*{\algorithmicindent}
\begin{algorithmic}[1]
\Repeat
\State Obtain a vertex solution $\vec{y}$ to (LP1).
\For{all $y_{ij}=b\in \{0,1\}$}
    \State Fix $y_{ij}=b$ and remove this variable from (LP1), updating the constraints as needed.
\EndFor
\Until{(LP1) is not modified.}
\For {each student $i$ with degree strictly more than $2$} 
 \State Add a capacity of one to an arbitrary school $j$ with $y_{ij} > 0$.
 \State Fix $y_{ij}=1$ and remove this student. \label{alg4-step:assign1}
 \State Decrease the capacity of $j$  correspondingly.
 \EndFor
 \State $ S_1 = \{ j | j \in \hat{T}, \mbox{degree}(j) > 2 \}$.
 \State $S_2 = \{ j | j \in \hat{T}, \mbox{degree}(j) = 2, C_j = 2\}$.
\For {each school $j \in S_1 \cup S_2 \cup (T' \setminus \hat{T})$}
\State $d = \mbox{degree}(j)$.
\State Create $d $ copies of $j$, each with capacity one. 
\State Assign (add an edge from) each $i$ with $y_{ij} > 0$ to a distinct copy of $j$.
\EndFor
\State If a school has degree one, reduce its capacity to one.
\State Set up the Cake Frosting instance as described in the text.
Let  $X$ be a perfect frosting.
\For{every student $i$}
\If {$[\frac{i-1}{r}, \frac{i}{r})\subseteq X$} 
\State Choose the edge from the even matching for student $i$.
\ElsIf {$[\frac{i-1}{r}, \frac{i}{r})\subseteq [0,1]\setminus X$} 
\State Choose the edge from the odd matching for student $i$.
\Else
\State Assign $i$ to some $j$ with $ y_{ij'} > 0$.\label{alg4-step:assign2}
\EndIf
\EndFor
\end{algorithmic}
\end{algorithm}
\fi

\begin{theorem}\label{thm:generalization-main2}
When the linear programming relaxation has a feasible solution, if $r$ is the number of constraints \cref{lp1:extra-constraints}, there is a $n^{O(r)}$ time algorithm that outputs an integer matching and achieves the following guarantee: 
\begin{itemize}
\item The constraints \cref{lp1:extra-constraints} are preserved up to an additive $O(r^2\cdot \qmax)$; and
\item The total violation in \cref{lp1:school} is $O(r^2)$.
\end{itemize}
\end{theorem}
\begin{proof}
We first argue about the violation in  \cref{lp1:extra-constraints}.  The only steps that affect the constraints are the assignment steps -- Steps~\ref{alg4-step:assign1} and \ref{alg4-step:assign2} in \cref{alg:overall-general}.  In Step~\ref{alg4-step:assign1}, the number of students assigned is $O(r)$ from \cref{eq:2g}, while that in Step~\ref{alg4-step:assign2} is $O(r^2)$. If these students are arbitrarily assigned, each assignment loses an additive $\qmax$ in the constraint. Therefore, the overall additive loss is $O(r^2\cdot \qmax)$. Note that the bound on the capacity violation follows from the proof of \cref{thm:main2} and holds even when these students are arbitrarily assigned.  
\end{proof}

Using Theorem 4.12 in~\cite{GrandoniRSZ14}, we can improve \cref{thm:generalization-main2} to the following corollary. We do this by guessing $8r^2/\epsilon$ chosen edges with highest utility for each group and subsequently applying \cref{alg:createpath,alg:frost}. Omitting the standard details yields the following corollary.
\begin{corollary}
\label{cor:ptas2}
Suppose the linear programming relaxation has a feasible solution and let $r$ be the number of constraints \cref{lp1:extra-constraints}. Then, for any constant $\epsilon > 0$, there is a $n^{O(r^3/\epsilon)}$ time algorithm that outputs an integer matching and that achieves the following guarantee: 
\begin{itemize}
\item The constraints \cref{lp1:extra-constraints} are preserved up to a multiplicative factor of $(1-\epsilon)$; and
\item The total violation in \cref{lp1:school} is $O(r^2)$.
\end{itemize}
\end{corollary}


\subsection{Better Bounds for Monotonic Constraints}

We next show that if the constraints $Q$ have an additional monotonicity structure, then we can generalize \cref{thm:main1} without the additive loss in the constraints. We say that $Q$ satisfies \textit{monotonicity} if for each student $i$, there is an ordering $\succeq_i$ of the schools $j_1\succeq_i j_2 \succeq_i \ldots \succeq_i j_m$ such that for all $\ell \in \{1,2,\ldots, r\}$ and $k \in \{1,2,\ldots, m-1\}$, we have $q_{ij_k}^{\ell} \ge q_{ij_{k+1}}^{\ell}$.

\begin{theorem}\label{thm:monotone1}
If the constraints $Q$ are monotone and the linear programming relaxation has a feasible solution, there is a polynomial-time algorithm that outputs an integer matching and achieves the following guarantee: 
\begin{itemize}
\item The constraints \cref{lp1:extra-constraints} are preserved; and
\item If each school is given one unit extra capacity, the total violation in \cref{lp1:school} over this is $2r$. 
\end{itemize}
\end{theorem}
\begin{theorem}\label{thm:monotone2}
If the constraints $Q$ are monotone, when the linear programming relaxation has a feasible solution, there is a $n^{O(r)}$ time algorithm that outputs an integer matching and that achieves the following guarantee: 
\begin{itemize}
\item The constraints \cref{lp1:extra-constraints} are preserved; and
\item The total violation in \cref{lp1:school} over this is $O(r^2)$.
\end{itemize}
\end{theorem}

\begin{proof}[Proof of \cref{thm:monotone1,thm:monotone2}]
We proceed as in \cref{thm:main1,thm:main2}. At each point where the algorithm assigns a student to $j^*=\arg\max_{j\in X} u_{ij}$ for some set $X$, we simply assign it to $\min_{k} \{ j_k \ | j_k\in X \}$.
That is, assign it to the most preferred school (according to $\succeq_i$).
This also preserves the $r$ constraints because of monotonicity.
\end{proof}

\subsection{Application: Weak Dominance of Ranks.}
\label{sec:dominance}
As a special case, we consider the setting in~\cite{asn22_optimalmatchings}. Here, every student ranks the schools it has an edge to, and this ranking may have ties. Let $r$ be the largest rank any student has, which can be much smaller than the number of schools. Given a matching, the rank of edge $(p,q)$ is the rank of school $q$ in student $p$'s ranking. A matching $M$ has \textit{signature} $\sigma=(\sigma_1,\sigma_2,\ldots,\sigma_r)$ if it has $\sigma_t$ rank $t$ edges for every $t\in [r]$. We say that signature $\sigma$ weakly dominates\footnote{The authors of \cite{asn22_optimalmatchings} use `cumulatively better than'.} signature $\rho$, or $\sigma\succ \rho$ if
\begin{align}
    \forall~t\in[r], \sum_{t'=1}^{t} \sigma_{t'} \geq \sum_{t'=1}^{t} \rho_{t'}.\label{eqn:rank}
\end{align}
Given an input signature $\rho$, the goal is to find a matching whose signature weakly dominates $\rho$. We term this the matching with ranking problem.
We have the following theorems, which directly follow from the observation that the constraints satisfy the monotonicity assumption. At each step where, for some set $X$, we assign $i$ to $\arg\max_{j\in X}{u_{ij}}$, we instead assign it to its most preferred school from $X$. This preserves the signature of any fractional assignment.
Our approach yields faster $n^{O(r)}$ time deterministic algorithms at the cost of small violations in capacities, whereas the algorithm in \cite{asn22_optimalmatchings} is randomized and takes $n^{O(r^2)}$ time. This concludes the proof of \cref{thm:ranking}.

%% file: experiments.tex
\section{ILP Benchmark and Simulation Study}
\label{sec:expt}
The goal of our simulation is to show the practicality of the convex programming framework in \cref{sec:convex} as well as our rounding methods for addressing group fairness constraints in assignments.

\paragraph*{ILP Benchmark.} First, note that our framework yields a benchmark for capacity violation for concave group fairness objectives. We first solve the convex program in \cref{sec:convex} to obtain the utility vector $\langle U_1^*, U_2^*, \ldots, U_g^* \rangle$. Subsequently, we can write an ILP to satisfy all utilities and violates total capacity the least as:
\[ \mbox{Minimize } \sum_{j \in T} \delta_j \]
\[ \begin{array}{rcll}
 \sum_{i \in S_k} \sum_{j \in T} u_{ij} y_{ij} & \ge & U^*_k & \forall k \\
\sum_{j \in T} y_{ij} & = & 1 & \forall i \in S \\
\sum_{i \in S} y_{ij} & \le & C_j + \delta_j & \forall j \in T \\ 
y_{ij} & \in & \{0,1\} & \forall i \in S, j \in T \\
\delta_j & \ge & 0 & \forall j \in T
\end{array}
\]

\cref{thm:main} says that the optimal value to this ILP is at most $\min \left(O(g^2), m + 2g \right)$. In our experiments, we compare the ILP benchmark for capacity violation with that of the rounding methods in \cref{alg:approx} and \cref{alg:frost}. We show that the capacity violation for both the ILP and the rounding methods is much smaller than the theoretical bounds in \cref{thm:main}, showing that group fairness functions have efficient near-optimal algorithms in the wild.
We now present our empirical results for the general school assignment problem. We present our experiments for the rank dominance problem in \cref{sec:expt2}.

\subsection{Empirical Results for School Assignment}
\paragraph*{Simulation Setup.} We generate $r = 100$ random instances with $n = 1000$ students, $m = 10$ schools with equal capacity $C = 100$, and $g = 7$ groups.  The instances are generated as follows. For every school $j$ and student $i$, an edge is added independently with probability $p = \frac{3}{m}$. Afterwards, edges are added from students with degree zero to a school chosen uniformly at random so that the minimum degree is $1$. Every school $j$ has a ``popularity measure'' $\alpha_j \sim {\tt Uniform} [0,1]$.
We set  $u_{ij}:=\hat{u}_{ij}\alpha_j$ where $\hat{u}_{ij} \sim {\tt Uniform}[0,1]$. This makes the utility of a school for different students correlated.
The capacities $C_j$  are set to minimize $\sum_j C_j$ so that all students can be feasibly assigned. This is found by solving an LP. Each group $k$ has a parameter $\beta_k \sim {\tt Uniform}[0,1]$. Each student belongs to group $k$ with probability $\beta_k$ independently of other students and its other group identities. We set the fairness objective to be Nash Welfare, corresponding to $f = \log$, which by \cref{eq:prop} achieves proportionality and its generalization to subsets of groups.

\paragraph*{Empirical Results.}  We first solve the convex program in \cref{sec:convex} to find the utility vector $\vec{U_k^*}$. We then consider the following three approaches to find an integer assignment with small capacity violations while preserving the utilities $\vec{U_k^*}$.

\begin{itemize}
    \item To find the integer assignment with minimum violation of capacities, $\sum_j \delta_j$, we solve the ILP described above. 
    \item We solve the LP in \cref{sec:convex} and round via \cref{alg:approx}.
    \item We solve the LP in \cref{sec:convex} and round via \cref{alg:frost}. 
\end{itemize} 


\begin{table}[]

\begin{center}
\begin{tabular}{@{}ccc@{}}
\toprule
Procedure           & Average violation     & Range of violations  \\ \midrule
Optimal             & 0.66                  & [0,1]                 \\ \midrule
GAP Rounding        & \multirow{2}{*}{2.3}  & \multirow{2}{*}{[0,6]} \\
(\cref{alg:approx}) &                       &                      \\ \midrule
Cake Frosting       & \multirow{2}{*}{1.24} & \multirow{2}{*}{[0,6]} \\
(\cref{alg:frost})  &                       &                      \\ \bottomrule
\end{tabular}
\caption{Results of experimental evaluation.}
\label{table}

\end{center}
\end{table}

The capacity violations are reported in \cref{table}.
For these instances, \cref{thm:main} implies an integer assignment violating capacities by at most $m + 2g = 24$. For all approaches above, the capacity violation is much lower than the theoretical bound, both on average and per instance, with our rounding schemes finding solutions very close to the ILP benchmark.  Further, both the ILP benchmark and \cref{alg:frost} run within a minute on a laptop on instances of this size. 
This is likely because most of the instances are already close to being integral. All the LP solutions had at most $30$ fractional variables, with an average of $21.73$.
This shows the practicality of the convex programming relaxation. 

%% file: conclusion.tex
\section{Conclusion}
\label{sec:conc}
We have presented a theoretically sound yet practical framework for handling group fairness and multi-objective optimization in capacitated assignment problems. Several open questions arise from our work. An immediate open question is to improve our theoretical bound on the capacity violation. We believe that a $O(g)$ violation should be possible in \cref{thm:main}. More broadly, our framework uses cardinal utilities and it would be interesting to incorporate group fairness into ordinal preferences, as in stable matchings. An even more basic question is to consider random allocations with ordinal preferences~\cite{MB2001,MB2002}, and define group fairness for lotteries over allocations. Finally, it would be interesting to incorporate group fairness into other optimization problems with rounding-based approximation algorithms, for instance, scheduling and routing problems.

%% file: hardness.tex
\section{Proofs of NP-Hardness results}
\label{app:hard}
\begin{proof}[Proof of \cref{thm:hardness_1}]
    We reduce from {\sc Set Cover} with a collection $\mathcal{C}$ of sets and a universe $U$ of elements. Suppose the goal is to decide if a set cover instance has $k$ sets that cover $U$. Then each element becomes a group and each set a student. A student belongs to a group if the corresponding set covers the corresponding element. There are two schools, $s_1$ and $s_2$. The former school has capacity $k$ and the latter has capacity $\infty$. Each student has utility $1$ for $s_1$ and $0$ for $s_2$. Then the goal of matching $k$ students to $s_1$ to give each group utility at least one is exactly the same as finding a set cover of size $k$, completing the proof. 
\end{proof}

\begin{proof}[Proof of \cref{thm:hardness-2}]
We reduce from {\sc Partition}. Given a set of numbers $x_1,..,x_n$, the goal is to decide if there is a subset of sum exactly $X/2$ where $X=\sum_i x_i$. For every number $x_i$, create two students $p_i$ and $q_i$, and one school $S_i$ of capacity 1. There is also a dummy school $S_0$ of capacity $n$. The students $p_i$ and $q_i$ have edges only to $S_i$ and $S_0$, where $p_i$ and $q_i$ have utility $x_i$ for $S_i$  and 0 for $S_0$. All the $p_i$ students belong to group 1 and all the $q_i$ students belong to group 2. We want to find a matching that gives utility $X/2$ to both groups, which is the proportional share.
 
Suppose there is a subset $T$ of the numbers that sums to exactly $X/2$. Then for every $x_i \in T$, we assign $p_i$ to $S_i$ and $q_i$ to $S_0$. For every $x_i\notin T$, we assign $q_i$ to $S_i$ and $p_i$ to $S_0$. Both groups get utility $X/2$ each. The reverse direction is similar, completing the proof.
\end{proof}

%% file: covering.tex
\section{Empirical Results for Weak Dominance of Ranks}
\label{sec:expt2}

\paragraph*{Simulation Setup.} We generate $100$ random instances with $n = 1000$ students, $m = 10$ schools with maximum rank $r=8$. The capacities $C_j$  are set so to minimize $\sum_j C_j$ so that all students can be feasibly assigned. This is found by solving an LP.
These instances are generated as in section~\ref{sec:expt}.  That is, for every school $j$ and student $i$, an edge is added independently with probability $p = \frac{3}{m}$. For every student, we select a random permutation of the schools in its neighborhood to obtain a ranking of the schools for that student. 

In the weak dominance of ranks setting, we generate the input signature $\sigma = (\sigma_1, \sigma_2, \ldots, \sigma_r)$ as follows: For $t \in [r]$, let $M_t$ denote a maximum matching on edges of rank up to $t$ in the generated instance.  We set $\sigma_1$ to be a random number between $0.9|M_1|$ and $|M_1|$. For, $i=2, \ldots, r$, to set  $\sigma_i$, we select a random number between $0$ and $|M_i| - \sum_{t'=1}^{i-1} \sigma_{t'}$.

\paragraph*{Empirical Results.} To decide whether there exists a feasible solution for the instance with the given signature, we solve the linear relaxation of the \ref{IP} defined in \cref{sec:results}. Out of the 100 instances, 79 of them admit a feasible solution, of which the solution is fractional in 26 instances.

Next, for the instances where the LP gives a fractional solution, we obtain an integral solution using the algorithm of \cref{thm:monotone1}. This yields capacity violations of at most 4, with an average violation of 0.53. Subsequently, we use the algorithm of Theorem~\ref{thm:monotone2} to obtain an integral solution. This yields capacity violation of at most 2 with an average violation of 0.07. As in the previous experiment, we observe that the violation values are much better than what the theoretical bounds predict.